\documentclass[conference]{IEEEtran}
\IEEEoverridecommandlockouts

\input{packages.tex}

\makeatletter
\renewcommand*\env@matrix[1][*\c@MaxMatrixCols c]{%
  \hskip -\arraycolsep
  \let\@ifnextchar\new@ifnextchar
  \array{#1}}
\makeatother

\newcommand{\cB}{\mathcal{B}}

\newcommand{\cS}{\mathcal{S}}
\newcommand{\cX}{\mathcal{X}}

\newcommand{\bu}{\mathbf{u}}

\newcommand{\E}{\mathbb{E}}

\newcommand\diag{\mathrm{diag}}

\newcommand{\bOne}{\mathbf{1}}

\newcommand{\NN}{\mathbb{N}}
\newcommand{\RR}{\mathbb{R}}
\newcommand{\CC}{\mathbb{C}}

\newcommand{\dd}{\mathrm{d}}

\newcommand{\DCSBM}{\mathrm{DCSBM}}

\newcommand{\TW}{\mathrm{TW}}

\newcommand{\Ai}{\mathrm{Ai}}

\theoremstyle{plain}

\newtheorem{theorem}{Theorem}[section]
\newtheorem*{theorem*}{Theorem}
\newtheorem{lemma}[theorem]{Lemma}
\newtheorem{proposition}[theorem]{Proposition}

\newtheorem{definition}[theorem]{Definition}

\newtheorem{conjecture}{Conjecture}

\theoremstyle{remark}

\newtheorem{remark}{Remark}

\usepackage{ulem}
    
\begin{document}

\title{Community Detection in High-Dimensional Graph Ensembles\\
\thanks{This work was partially supported by Dept of the Air force grant FA8650-19-C-1712 and Army Research Office grant  W911NF2310343.}
}

\author{\IEEEauthorblockN{Robert Malinas, Dogyoon Song, Alfred O. Hero III}
\IEEEauthorblockA{\textit{Department of Electrical \& Computer Engineering} \\
\textit{University of Michigan}\\
Ann Arbor, MI, 48105, USA \\
Email: \{rmalinas, dogyoons, hero\}@umich.edu}
}

\maketitle

\begin{abstract}
Detecting communities in high-dimensional graphs can be achieved by applying random matrix theory where the adjacency matrix of the graph is modeled by a Stochastic Block Model (SBM). However, the SBM makes an unrealistic assumption that the edge probabilities are homogeneous within communities, i.e., the edges occur with the same probabilities. The Degree-Corrected SBM is a generalization of the SBM that allows these edge probabilities to be different, but existing results from random matrix theory are not directly applicable to this heterogeneous model. In this paper, we derive a transformation of the adjacency matrix that eliminates this heterogeneity and preserves the relevant eigenstructure for community detection. We propose a test based on the extreme eigenvalues of this transformed matrix and (1) provide a method for controlling the significance level, (2) formulate a conjecture that the test achieves power one for all positive significance levels in the limit as the number of nodes approaches infinity, and (3) provide empirical evidence and theory supporting these claims.
\end{abstract}


\section{Introduction}
Networks are composed of nodes and pairs of nodes, known as edges or connections, and are capable of representing a diverse range of data. 
For instance, social networks consist of nodes representing users and edges denoting interpersonal relationships \cite{zachary1977information}. 
In satellite communication networks, satellites and microwave channels comprise the nodes and edges, respectively \cite{elbert2008introduction}. 
Due to their versatility, statistical inference with network data is important to a variety of scientific studies such as social science \cite{wasserman1994social}, metabolism \cite{jeong2000large}, and epidemiology \cite{moore2000epidemics} (see \cite{newman2003structure} for a review). 

Networks often contain communities---groups of nodes within which connections are especially dense \cite{newman2004finding}. 
Determining whether or not a network contains communities, which is known as community detection, is a fundamental problem in network data analysis \cite{fortunato2010community}.
Methods for community detection typically fall into one of three categories \cite{zhao2012consistency}: greedy algorithms \cite{newman2004finding}, global optimization methods \cite{shi2000normalized, ng2001spectral, newman2006modularity}, or probabilistic models \cite{holland1983stochastic, airoldi2008mixed, karrer2011stochastic, abbe2018community, newman2007mixture}. 

In this work, we address community detection within a probabilistic model using the framework of statistical hypothesis testing. 
The Erdős–Rényi Model (ERM) \cite{erdos1959random} is perhaps the simplest probabilistic network model, in which $n$ nodes are connected independently at random with some probability $p \in (0,1)$; however, modeling multiple communities requires additional degrees of freedom. 
A generalization of the ERM, the Stochastic Block Model (SBM) \cite{holland1983stochastic} consists of $n$ nodes, each of which belongs to one of $1 \leq k \leq n$ communities, a community assignment function $\varphi:[n] \rightarrow [k]$, and an irreducible symmetric matrix of probabilities $P \in (0,1)^{k \times k}$. 
One typically thinks of $k$ as being much smaller than $n$. 
A network is drawn from the SBM if any pair of nodes $i, j \in [n]$ is connected with probability $(P)_{\varphi(i) \varphi(j)} = (P)_{\varphi(j)\varphi(i)}$ mutually independently from all other node pairs. 
Although simple, the SBM provides a capable sandbox for describing many interesting problems, e.g., planted partition (or clique) detection \cite{condon2001algorithms}, and exhibits interesting 
phase transitions (see \cite{abbe2018community} for a recent survey). 
The hypotheses for community detection in SBMs is straightforward: the null hypothesis is that $k=1,$ i.e., the SBM is an ERM. The alternative hypothesis is that $k > 1.$

Despite its merits, the SBM has limitations in capturing real-world network features like the rich-club phenomenon, where high-degree nodes tend to connect, even across different communities \cite{zhou2004rich, colizza2006detecting}.
More generally, the SBM is too simplistic to represent high levels of degree variation within a single community \cite{lancichinetti2008benchmark}, leading to underfitting in many real networks. 
This underfitting increases the type I error rate, i.e., the chance of incorrectly declaring that more than one community is present, in the context of statistical hypothesis testing for community detection \cite{jin2023phase}. 

To remedy this underfitting, the authors of \cite{karrer2011stochastic} proposed the Degree-Corrected Stochastic Block Model (DCSBM). 
This generalization of the SBM includes additional parameters to capture node-specific connection affinity: the likelihood for a node to connect to other nodes regardless of community membership (cf. Definition \ref{def:DCSBM}). 
However, while offering improved modeling capabilities, working with this more complex model presents significant mathematical challenges. 

This paper proposes a statistical test, based on random matrix theory, for community detection within a large DCSBM. 
Our key contributions are summarized here:
\begin{enumerate}
    \item
    We introduce a transformation method that converts the heterogeneous and intractable adjacency matrix into a tractable Wigner ensemble. Additionally, we present an approach for directly estimating the transformed matrix from a single snapshot of the adjacency matrix.
    \item 
    We propose a method to determine the significance level of the test, which is exact in the limit that $n$ approaches infinity.
    \item
    We argue that our test achieves perfect detection for all positive significance levels as $n$ tends to infinity, and provide empirical evidence of this claim.
\end{enumerate}

In Section \ref{sec:preliminaries}, we provide a formal definition of the DCSBM and identify the main problem addressed by this paper.
In Section \ref{sec:relatedwork}, we overview related work. 
In Section \ref{sec:analysis}, we initiate our analysis, providing relevant results from random matrix theory when necessary. 
In Section \ref{sec:test}, we propose the test, provide a method for estimating its false alarm rate, and argue for its consistency in the limit as the number of nodes approaches infinity. 
In Section \ref{sec:empiricalresults}, we provide empirical evidence supporting the assertions of Section \ref{sec:test}.
In Section \ref{sec:conclusion}, we conclude and suggest directions for future study.

\section{Preliminaries}\label{sec:preliminaries}
We begin with a formal definition of the Degree-Corrected Stochastic Block Model (DCSBM). 
\begin{definition}[Degree-Corrected Stochastic Block Model]\label{def:DCSBM}
Let $n, k \in \NN$ such that $k < n$, and let $\epsilon \in (0, 1/2]$. 
Furthermore, let 
\begin{itemize}
    \item $\varphi:[n] \rightarrow [k]$ be surjective;
    \item $W = (w_{\mu\nu})_{\mu, \nu = 1}^k \in \RR_+^{k \times k}$ be symmetric and irreducible;
    \item $\boldsymbol{\theta} = (\theta_1, \ldots, \theta_n) \in (0, 1]^n$;
\end{itemize}
such that 
\begin{itemize}
    \item $\displaystyle \sum_{i \in \varphi^{-1}(\{\mu\})}\theta_i = 1$ for all $\mu \in [k]$;
    \item $w_{\varphi(i) \varphi(j)} \theta_i \theta_j \in [\epsilon,1 - \epsilon]$ for all $i, j \in [n].$
\end{itemize}

A random matrix $A \in \{0, 1\}^{n \times n}$ is an adjacency matrix drawn from the Degree-Corrected Stochastic Block Model with parameter $(\epsilon, \varphi, \boldsymbol{\theta}, W)$, written $A \sim \DCSBM_n(\epsilon, \varphi, \boldsymbol{\theta}, W)$, if
\begin{enumerate}[label=(\roman*)]
\item $A$ is symmetric;
\item $\{A_{ij}\}_{1 \leq i \leq j \leq n}$ is a mutually independent set of random variables;
\item for all $i, j \in [n]$, if $i \in \varphi^{-1}( \{ \mu \})$ and $j \in \varphi^{-1}(\{ \nu \})$, then\footnote{A random variable $X$ is Bernoulli with parameter $p$, $p \in [0,1],$ which we denote by $X \sim \mathrm{Bern}(p)$, if $X = 1$ with probability $p$ and $X = 0$ otherwise.} $(A)_{ij} \sim \mathrm{Bern} \left ( \theta_i \theta_j w_{\mu \nu} \right )$.
\end{enumerate}
\end{definition}
The symbols $n$ and $k$ will be used throughout the rest of this paper to refer to the number of nodes and the number of communities, respectively, in a DCSBM model.
Interpretations of the parameters $\boldsymbol{\theta}$ and $W$ may be found in \cite{karrer2011stochastic}.
We aim to test the hypotheses
\begin{align*}
H_0 & : k = 1, \\
H_1 & : k > 1.
\end{align*}
The crux of this paper is that the model under $H_0$ is not, in general, an ERM; rather, it is a DCSBM with 1 community. 
The expected degree distribution of a one-community DCSBM may be arbitrary, whereas the ERM expected degree distribution is always a single atom with mass one.

\section{Related Work} \label{sec:relatedwork}
We overview related work on community detection methods that are based on matrices such as the adjacency, modularity, or graph Laplacian matrices, within a DCSBM framework.

\paragraph{Signed Polygon Statistics}
A class of Signed Polygon statistics proposed in \cite{jin2019optimal} assigns scores to each $m$-gon in a network for some $m \geq 3.$ 
The scores are based on the degree of each node in the $m$-gon, and the statistic is the sum of all such scores. 
The detection performance of Signed Polgyon statistics is shown to be robust to sparsity and mixed membership, in which communities may overlap. 
Furthermore, these statistics are capable of detecting small planted cliques with size on the order $n^{1/2}$ \cite{jin2023phase}.
\paragraph{Spectral Clustering}
The authors of \cite{jin2015fast} consider entrywise ratios of a small number of leading eigenvectors of the adjacency matrix. It is argued that these ratio vectors effectively mod out community-independent degree heterogeneity, and clustering the $(k-1)$-tuple of ratios for each node via k-means is capable of detecting communities. 
The authors of \cite{tiomokoali2018improved} characterize the empirical spectral distributions of a class of normalized modularity matrices drawn from the DCSBM. In addition to describing phase transitions for community detection, they propose a spectral clustering algorithm that finds an optimal normalization for the modularity matrix, followed by a k-means clustering of its eigenvectors.

\paragraph{Extreme Eigenvalue Tests for Special DCSBMs}

Perhaps the first work that uses the extreme (the largest and smallest) eigenvalues for community detection, \cite{nadakuditi2012graph} establishes a phase transition in the largest eigenvalue of the modularity matrix for a special case of the DCSBM, namely the planted partition model.
The authors of \cite{bickel2016hypothesis} also test the extreme eigenvalues of a transformed adjacency matrix. The transformation is similar to that of this paper in that they are both entrywise centerings and rescalings of the adjacency matrix; however, \cite{bickel2016hypothesis} does not consider a general DCSBM null hypothesis, which is the main focus of this paper.
Finally, a goodness-of-fit test is proposed in \cite{lei2016goodness} in which $k = k_0$ vs $k \neq k_0$ is tested sequentially for $k_0 = 1, 2, \ldots$ until the $k = k_0$ hypothesis is accepted. 
This method can be adapted to community detection by terminating the sequence after testing $k_0 = 1$; however, like \cite{bickel2016hypothesis}, \cite{lei2016goodness} does not characterize the significance level under a DCSBM null hypothesis.

\section{Community Detection for the DCSBM Using Random Matrix Theory} \label{sec:analysis}
Let $n, \ k, \ \varphi, \ \boldsymbol{\theta} = (\theta_1, \ldots, \theta_n),$ and $W = (w_{\mu\nu})_{\mu, \nu = 1}^k$ be as in Definition \ref{def:DCSBM}, and suppose $A = (a_{ij})_{i, j = 1}^n \sim \DCSBM_n(\epsilon, \varphi, \boldsymbol{\theta}, W)$. We may write
\[
A = \E A + (A - \E A).
\]
\subsection{Analysis Under One-Community DCSBM Null Hypothesis}
Under the null hypothesis $H_0,$
\[
\E A = w_{11} \boldsymbol{\theta \theta}^T,
\]
hence, $A$ is the sum of a rank-one matrix $\E A$ and a centered random matrix $A - \E A$. 
Moreover, the entries in the diagonal and upper triangle $\{ a_{ij}: 1 \leq i \leq j \leq n \}$ are mutually independent.
The difficulty with directly analyzing $A$ under the DCSBM is that the variances
\[
s_{ij} \coloneqq \E \left \vert a_{ij} - \E a_{ij} \right \vert^2 = w_{11} \theta_i \theta_j (1 - w_{11} \theta_i \theta_j), \ \ \forall i, j \in [n],
\]
are, in general, heterogeneous in that they depend on $i,j$. 
Due to this heterogeneity, the eigenvalues of $A - \E A$ can only be described implicitly.
Specifically, the empirical spectral distribution of $A - \E A$ (cf. Definition \ref{defn:esd}) is asymptotically characterized via its Stieltjes transform as the implicit solution to a quadratic vector equation \cite{ajanki2017universality, ajanki2019quadratic}, the fundamental properties of which have only recently been described. 
Even less is known about finite-rank perturbations of such a matrix, making it difficult to describe the extreme eigenvalues of $A$.

\subsection{Special Case with Homogeneous Parameters}
In the special case of homogeneous parameters where $\theta_i = \theta_j$ for all $i, j \in [n]$, the DCSBM model under the null hypothesis reduces to an ERM with parameter $w_{11} \theta_1^2.$ In this special case, $A - \E A$ is distributed as a scaled Wigner ensemble \cite{wigner1951statistical, wigner1955characteristic, wigner1958distribution}, defined below.
\begin{definition}[Wigner ensemble {\cite{baisilversteinbook}}]\label{defn:wigner}
    Let $H = (h_{ij})_{i,j=1}^n \in \CC^{n \times n}$ be a Hermitian random matrix. 
    The random matrix ensemble $H$ is a Wigner ensemble if its entries are centered and normalized, i.e., $\E h_{ij} = 0$ and $\E |h_{ij}|^2 = \frac{1}{n}$ for all $i, j \in [n]$, and its entries in the diagonal and upper triangle $\{ h_{ij}: 1 \leq i \leq j \leq n \}$ are mutually independent.
\end{definition}
Much is known about Wigner ensembles and, in particular, their extreme eigenvalues and the extreme eigenvalues of low-rank perturbations thereof. Therefore, it is desirable to work with Wigner ensembles when possible. To this end, we define a map that transforms a DCSBM adjacency matrix $A$ to a (scaled) Wigner matrix $B$.
\begin{proposition}
Let $n, \ k, \ \varphi, \ \boldsymbol{\theta} = (\theta_1, \ldots, \theta_n),$ and $W = (w_{\mu\nu})_{\mu, \nu = 1}^k$ be as in Definition \ref{def:DCSBM}, and suppose $A \sim \DCSBM_n(\varphi, \boldsymbol{\theta}, W)$. Furthermore, let $B \in \RR^{n \times n}$ such that
\[
(B)_{ij} = \dfrac{(A)_{ij} - \theta_i \theta_j w}{\sqrt{ \theta_i \theta_j w (1 - \theta_i \theta_j w) } }, \qquad \forall i, j \in [n],
\]
where $w \equiv w_{11}.$
If $k = 1$, then $n^{-1/2} \cdot B$ is a Wigner ensemble.
\end{proposition}
\begin{proof}
It follows from (i) and (ii) of Definition \ref{def:DCSBM} that $B$ is symmetric with independent entries in the diagonal and upper triangle. The proof is complete by noticing that for any $i, j \in [n],$
\begin{enumerate}
\item 
\begin{align*}
    \E \left [ n^{-1/2}\cdot (B)_{ij} \right ] & = n^{-1/2} \cdot \dfrac{\E(A)_{ij} - \theta_i \theta_j w}{\sqrt{ \theta_i \theta_j w (1 - \theta_i \theta_j w) } } \\
    & = n^{-1/2} \cdot \dfrac{\theta_i \theta_j w - \theta_i \theta_j w}{\sqrt{ \theta_i \theta_j w (1 - \theta_i \theta_j w) } } \\
    & = 0,
\end{align*}
and
\item 
\begin{align*}
\E \left \vert n^{-1/2 } \cdot (B)_{ij} \right \vert^2 & = \dfrac{ \E \left \vert (A)_{ij} - \theta_i \theta_j w \right \vert^2}{n \cdot \theta_i \theta_j w (1 - \theta_i \theta_j w)} \\
& = \dfrac{ \theta_i \theta_j w - (\theta_i \theta_j w)^2}{n \cdot \theta_i \theta_j w (1 - \theta_i \theta_j w)} \\
& = \dfrac{1}{n}.
\end{align*}

\end{enumerate}
\end{proof}
\subsection{Properties of Wigner Ensembles}
We now discuss some important results from random matrix theory. Following the seminal work of Baik, Ben Arous, and Péché \cite{baik2005phase}, which studies the extreme eigenvalues of `spiked' sample covariance matrices \cite{johnstone2001distribution}, significant attention was devoted to developing a parallel theory for low-rank perturbations of Wigner ensembles \cite{peche2006largest, bloemendal2012limits, pizzo2013finite, knowles2013isotropic, knowles2014outliers, lee2014necessary, bloemendal2016limits}. The key relevant results of this line of work, under mild technical conditions, include:
\begin{enumerate}
\item the distributions of the extreme eigenvalues;
\item a sharpening of the empirical spectral distribution asymptotics (cf. Definition \ref{defn:esd}) from \cite{wigner1951statistical, wigner1955characteristic, wigner1958distribution};
\item a phase transition for the extreme eigenvalues under an additive low-rank perturbation. 
\end{enumerate}
Next, we provide the specific results required for our community detection analysis.

\subsubsection{Distribution of the Extreme Eigenvalues}

Under mild technical conditions, the marginal distributions of the extreme eigenvalues of an appropriately centered-and-rescaled Wigner ensemble each converge weakly (cf. {\cite[Section 25]{billingsley1995probability}}) to the Tracy-Widom distribution, defined next.

\begin{definition}[Tracy-Widom ensemble \cite{tracy1994level, tracy1996orthogonal}] \label{def:TWensemble}
The Tracy-Widom distribution $\TW_1$ is the probability measure on $(\RR, \mathcal{B}(\RR))$ with CDF
\[
F_1(x) = \exp \left [ -\dfrac{1}{2} \int_{x}^\infty q(y) + (y - x) q^2(y) \dd y \right ],
\]
where $q(y)$ is the unique solution to the Painlevé II differential equation
\[
\dfrac{\dd^2}{\dd y^2}q(y) = y q(y) + 2 q^3(y),
\]
satisfying the boundary condition $q(y) \asymp \Ai(y)$ as $y \rightarrow \infty$, where $\Ai(y)$ denotes the Airy function of the first kind.
\end{definition}

The Tracy-Widom density, i.e., $\frac{\dd F_1}{\dd x}$ is plotted in red in Figure \ref{fig:TW1}.

\begin{theorem}[{\cite[excerpted from Theorem 1.2]{lee2014necessary}}]
For each $n \in \NN,$ let $H_n$ be an $n \times n$ Wigner ensemble. Suppose that
\begin{equation}\label{eq:tech_cond_tw}
\lim_{s \rightarrow +\infty} s^4 \cdot P \left ( \left \vert n^{1/2} (H)_{12} \right \vert \geq s \right ) = 0.
\end{equation}
Then
\[
P \left ( n^{2/3} \cdot (\lambda_1(H_n) - 2) \leq x \right ) \rightarrow F_1(x), \qquad \forall x \in \RR.
\]
Moreover, the logical inverse holds, i.e., \eqref{eq:tech_cond_tw} is necessary. A similar result holds for $-\lambda_{n}(H_n)$.
\end{theorem}
Trivially, $n^{-1/2}\cdot B$ satisfies \eqref{eq:tech_cond_tw} because its entries are bounded almost surely.

\subsubsection{Bulk Characterization of the Spectrum}

Other key results offer characterizations of the bulk of the spectrum of a Wigner ensemble, i.e., the eigenvalues outside of an epsilon neighborhood of the extreme eigenvalues.
Namely, the empirical spectral distribution, defined below, of a Wigner matrix converges almost surely to the semicircle law. 
\begin{definition}[Empirical spectral distribution]\label{defn:esd}
    Let $H \in \CC^{n \times n}$ be a Hermitian matrix. 
    The empirical spectral distribution (ESD) of $H$ is the probability measure
    \begin{equation}\label{eqn:esd}
        \mu_H \coloneqq \dfrac{1}{n} \sum_{i = 1}^n \delta_{\lambda_i(H)},
    \end{equation}
    where $\left (\lambda_i(H) \right )_{i = 1}^n$ is the multiset of eigenvalues of $H$, and $\delta_x$ is the Dirac measure with support $\{x\}$ for any $x \in \RR$.
\end{definition}
Define 
\[
\mu_{sc}(E) = \int_{E} \rho_{sc}(x) \dd x, \qquad \forall E \in \mathcal{B}(\RR),
\]
where $\rho_{sc}(x) \coloneqq \dfrac{1}{2 \pi} \sqrt{4 - x^2} \mathbbm{1}_{[-2, \hspace{2pt} 2]}$ for all $ x \in \RR.$
\begin{theorem}[{\cite[Theorem 2.5]{baisilversteinbook}}]
 For each $n \in \NN,$ suppose that $H_n$ is an $n \times n$ Wigner ensemble. Then, with probability one, 
 \[
\mu_{H_n} \rightarrow \mu_{sc}
\]
weakly.
\end{theorem}

\subsection{Estimating the Transformed Matrix $B$}\label{sec:estimate_B}
In section \ref{sec:analysis}, we introduced a transformed adjacency matrix $B$ and proved that it is a Wigner ensemble under the null hypothesis; however, $B$ depends on the unknown quantities $\boldsymbol{\theta}$ and $w_{11}.$ Therefore, we must estimate $B$ to form a viable statistical test. We propose the estimator
\[
(\hat{B})_{ij} = \dfrac{(A)_{ij} - \frac{(A \bOne \bOne^T A)_{ij}}{\bOne^T A \bOne}}{\sqrt{\frac{(A \bOne \bOne^T A)_{ij}}{\bOne^T A \bOne} \left (1 - \frac{(A \bOne \bOne^T A)_{ij}}{\bOne^T A \bOne} \right )}}, \qquad \forall i, j \in [n],
\]
informed by the following theorem. 
\begin{theorem}
\label{thm:entrywise_estimator}
Let $n, \ k, \ \varphi, \ \boldsymbol{\theta} = (\theta_1, \ldots, \theta_n),$ and $W = (w_{\mu\nu})_{\mu, \nu = 1}^k$ be as in Definition \ref{def:DCSBM}, and suppose $A = (a_{ij})_{i, j = 1}^n \sim \DCSBM_n(\epsilon, \varphi, \boldsymbol{\theta}, W)$. If $k = 1$, then 
\begin{align*}
P & \left ( \left \vert  \left ( \dfrac{A \bOne \bOne^T A}{ \bOne^T A \bOne} \right )_{ij} - w \theta_i \theta_j \right \vert \leq \dfrac{8}{\epsilon} \left ( \sqrt{t} n^{-1/2} + t n^{-1} \right ) \right ) \\
& \geq 1 - 2 \left ( 2 e^{-2 t^2} + e^{-\frac{n^2 \epsilon^2}{18}} \right ), \qquad \forall t > 0, \ \forall i, j \in [n],
\end{align*}
where $w \equiv w_{11}.$
\end{theorem}
\begin{proof}
See Appendix (Section \ref{sec:proof_main}).
\end{proof}
To summarize, $\frac{(A \bOne \bOne^T A)_{ij}}{\bOne^T A \bOne}$ fluctuates around the unknown parameters $w_{11} \theta_i \theta_j$ at the scale $n^{-1/2}$ for any $i, j \in [n].$ 
Moreover, the fluctuation is subgaussian. 
From these facts, we conjecture the following.

\begin{conjecture}\label{conj:statistic}
Let $n, \ k, \ \varphi, \ \boldsymbol{\theta} = (\theta_1, \ldots, \theta_n),$ and $W = (w_{\mu\nu})_{\mu, \nu = 1}^k$ be as in Definition \ref{def:DCSBM}, and suppose $A = (a_{ij})_{i, j = 1}^n \sim \DCSBM_n(\epsilon, \varphi, \boldsymbol{\theta}, W)$.
Then if $k = 1$,
\[
\left \vert P \left ( n^{2/3} \cdot  \left ( \lambda_{1} \left ( n^{-1/2} \cdot \hat{B} \right ) - 2 \right ) \leq x \right ) - F_1(x) \right \vert = o(1),
\]
and
\[
\left \vert P \left ( n^{2/3} \cdot  \left ( - \lambda_{n} \left ( n^{-1/2} \cdot \hat{B} \right ) - 2 \right )  \leq x \right ) - F_1(x) \right \vert = o(1),
\]
for all $x \in \RR$.
\end{conjecture}

In other words, the distributions of the extreme eigenvalues of $\hat{B}$ are the same as those of a Wigner ensemble asymptotically. 

\section{Test Statistic}\label{sec:test}
We propose the statistic
\begin{align*}
T = \max \left \{ n^{2/3} \cdot  \left ( \lambda_{1} \left ( n^{-1/2} \cdot \hat{B} \right ) - 2 \right ) \right., & \\
\left. \ n^{2/3} \cdot  \left ( - \lambda_{n} \left ( n^{-1/2} \cdot \hat{B} \right ) - 2 \right ) \right \}.
\end{align*}

This statistic is reasonable because
\begin{enumerate}
\item under the null hypothesis, given Conjecture \ref{conj:statistic}, $T$ fluctuates around $0,$ following a $\TW_1$ distribution asymptotically;
\item under the alternative hypothesis, we expect that $\left \vert \lambda_1 \left ( n^{-1/2} \cdot \hat{B} \right ) - 2\right \vert \gg n^{-2/3}$ for a large class of alternative models $W,$ $\boldsymbol{\theta},$ and $\varphi,$ with a similar statement holding for $-\lambda_n \left ( n^{-1/2} \cdot \hat{B} \right ).$
\end{enumerate}
The latter point comes from a BBP-type phase transition for low-rank perturbations of Wigner ensembles, e.g., {\cite[Theorem 2.7]{knowles2013isotropic}}. In section \ref{sec:empiricalresults}, we provide empirical evidence that indicates a threshold test based on $T$ attains power one in the limit as $n \rightarrow \infty$ for all positive significance levels, i.e., the test is asymptotically consistent. We leave a thorough power analysis for future work.

\subsection{False Alarm Rate}

Based on Conjecture \ref{conj:statistic}, we pose the following conjecture on the false alarm rate of the proposed statistic $T$.
\begin{conjecture} \label{conj:far}
Let $G_1 = 1 - F_1$. If $k = 1,$ then
\[
P \left ( T \geq G^{-1}_1(\alpha/2) \right ) \lesssim \alpha, \qquad \forall \alpha \in (0, 1).
\]
\end{conjecture}
\begin{remark}
We note that $G_1^{-1}$ exists because $F_1$ is strictly monotonic: the $\TW_1$ distribution is absolutely continuous.
\end{remark}
In summary, rejecting the null hypothesis if $T \geq G^{-1}_1(\alpha/2)$ yields a type I error of at most $\alpha,$ for any $\alpha \in (0,1),$ at least asymptotically. 
The argument $\alpha/2$ of the quantile function $G^{-1}_1$ is due to Bonferonni correction, because we are simultaneously testing two eigenvalues.

\section{Empirical Results}\label{sec:empiricalresults}

In this section, we present empirical evidence supporting the assertion that the eigenvalues of $n^{-1/2} \cdot \hat{B}$ behave like those of a Wigner ensemble. For all simulations, we generate $\boldsymbol{\theta}$ via
\[
\theta_i =  \dfrac{X_i}{\sum_{j = 1}^n X_{j}}, \qquad \forall i \in [n],
\]
where  $X_j \sim \mathrm{Unif}[0.1, 0.9],$ for all  $j \in [n].$
For Figure \ref{fig:ROCs}, we set 
\[
W = 
D \cdot \begin{pmatrix} 0.4 & 0.2 & 0.2 \\0.2 & 0.6 & 0.3 \\0.2 &0.3 & 0.5 \end{pmatrix} \cdot D,
\]
where
\[
D = \diag \left ( \displaystyle \sum_{j \in \varphi{-1}(\{1\})}^n X_{j}, \displaystyle \sum_{j \in  \varphi{-1}(\{2\})}^n X_{j}, \displaystyle \sum_{j \in \varphi{-1}(\{3\})}^n X_{j} \right ),
\]
and
\[
\left \vert \varphi^{-1}(\{ \mu \}) \right \vert = \dfrac{n}{3}, \qquad \forall \mu \in [3].
\]
For Figures \ref{fig:semicirclelaw} and \ref{fig:TW1}, we set $w_{11} = \frac{1}{2} \sum_{j = 1}^n X_j.$

Figures \ref{fig:semicirclelaw} and \ref{fig:TW1} support that $n^{-1/2} \cdot \hat{B}$ has a spectrum close to that of a Wigner ensemble. In particular, Figure \ref{fig:TW1} indicates that the extreme eigenvalues of $n^{-1/2} \cdot \hat{B}$ converge to those of a Wigner ensemble at a rate faster than $n^{-2/3}.$ Figure \ref{fig:ROCs} offers a glimpse into the asymptotic power of the proposed test, indicating that it is asymptotically one for any positive significance level.

\begin{figure}[h]
\centering
\includegraphics[width=0.8\linewidth]{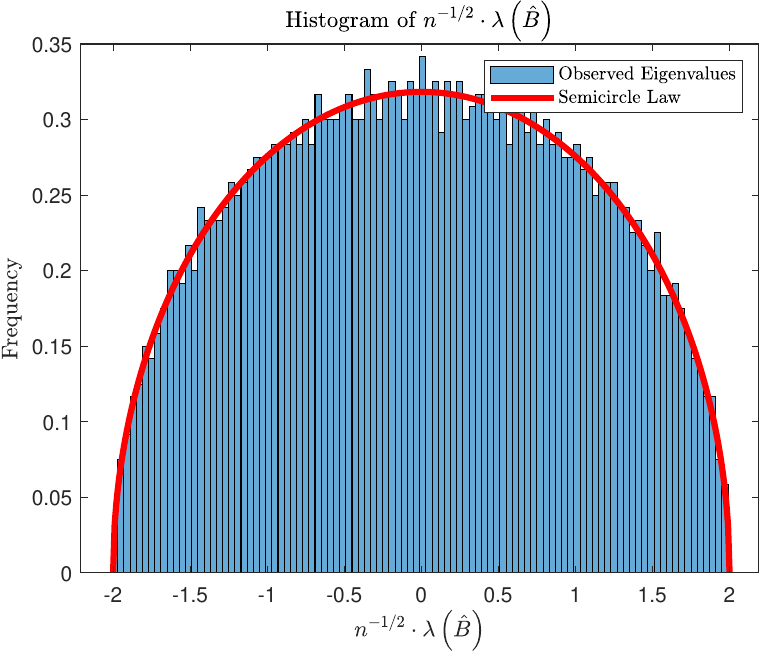}
\caption{A histogram of the eigenvalues of a single realization of $n^{-1/2} \cdot \hat{B}$ under the null hypothesis $k = 1,$ with $n = 3000.$ Overlaid on the histogram is the semicircle density $\rho_{sc}(x) = \dfrac{1}{2 \pi} \sqrt{4 - x^2} \cdot \mathbbm{1}_{[-2, \hspace{2pt} 2]}.$}
\label{fig:semicirclelaw}
\end{figure}

\begin{figure}[h]
\centering
\includegraphics[width=0.8\linewidth]{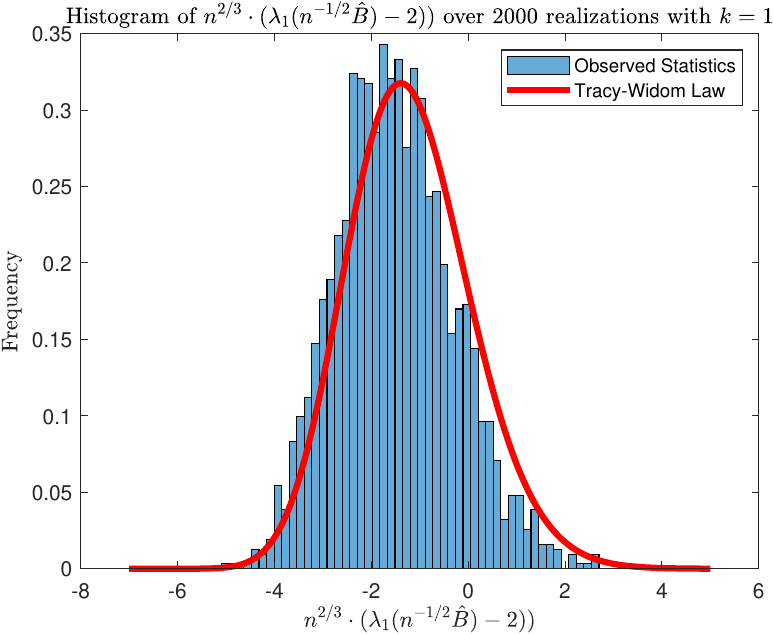}
\caption{A histogram of  $n^{2/3} \cdot  \left ( \lambda_{1} \left ( n^{-1/2} \cdot \hat{B} \right ) - 2 \right )$ over 2000 independent realizations with $n = 500$. Overlaid on the histogram is the $\TW_1$ density $\dfrac{\dd F_1}{\dd x},$ computed in software using \cite{chiani2014distribution}.}
\label{fig:TW1}
\end{figure}

\begin{figure}[h]
\centering
\includegraphics[width=0.8\linewidth]{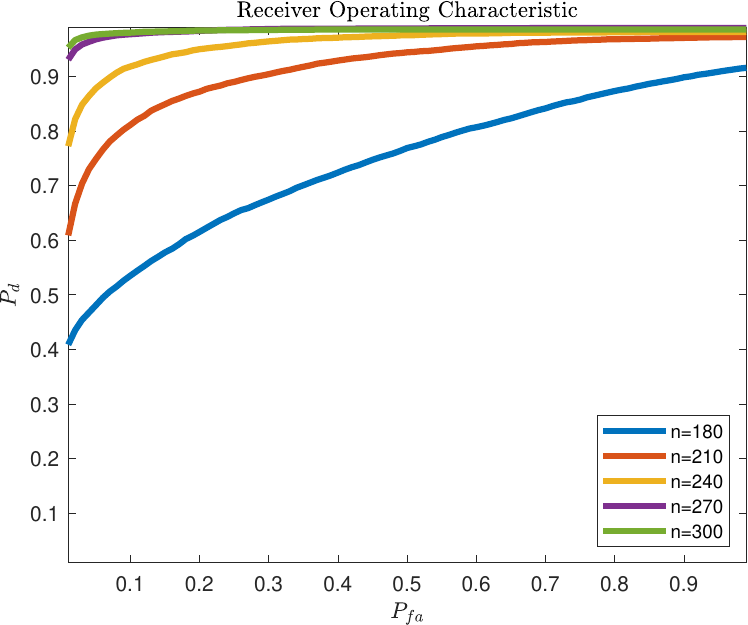}
\caption{Receiver operating characteristics for different values of $n$.}
\label{fig:ROCs}
\end{figure}

\section{Conclusion} \label{sec:conclusion}

In this paper, we proposed a test for detecting communities within a Degree-Corrected Stochastic Block Model. The test is based on the extreme eigenvalues of a an element-wise centered and rescaled adjacency matrix. Roughly, the proposed centering and rescaling are consistent with a transformation that maps the adjacency matrix $A$ to a Wigner ensemble. 
Because of this, we are able to approximate the distribution of the proposed statistic using the Tracy-Widom distribution: the asymptotic distribution of the extreme eigenvalues of a Wigner ensemble. Additionally, we provided a method for controlling the false alarm rate of the proposed test. Future work includes an analysis of the power of this test, which we believe converges as $n\rightarrow \infty$ to one for any positive significance level, and rigorous proofs of Conjectures \ref{conj:statistic} and $\ref{conj:far}.$

\section{Appendix}

\subsection{Notation and Conventions}\label{sec:notation}
Let $\NN$ be the set of positive integers. For $n \in \NN$, let $[n]:= \{1, \dots, n\}$. 
Let $\CC$ be the set of complex numbers, let $\RR$ be the set of real numbers, and define $\RR_+ \coloneqq (0, \infty].$ 
For a matrix $A \in \CC^{m \times n},$ we write $(A)_{ij}$ for the element in the $i^{th}$ row and $j^{th}$ column of $A$. 
Let $\bOne_n \coloneqq (1, 1, \ldots, 1) \in \RR^n.$
For $i, j \in [n],$ we set $\delta_{ij} = 1$ if $i = j$ and $\delta_{ij} = 0$ otherwise.
For $n \in \NN,$ we reserve the letter $I_n \in \CC^{n \times n}$ to denote the identity matrix throughout this paper.
For a diagonalizable matrix $A \in \CC^{n \times n}$, we write $\left (\lambda_i(A)\right )_{i = 1}^n$ for the multiset of eigenvalues of $A$ such that $\lambda_1(A) \geq \lambda_2(A) \geq \cdots \geq \lambda_n(A)$.
For a set $\cS \subseteq \CC$, we use the notation $A = (a_{ij})_{i, j = 1}^n \in \cS^{n \times n}$ to denote the matrix $A$ with elements $(A)_{ij} = a_{ij} \in \cS$ for all $i, j \in [n].$
For a vector $\bu = (u_1, \ldots, u_n) \in \CC^n,$ we let $\diag(\bu) \in \CC^{n \times n}$ denote the matrix with elements $\left ( \diag(\bu) \right )_{ij} = \delta_{ij} u_i.$ For a topological space $\cX$, we use $\cB(\cX)$ to denote the Borel $\sigma-$algebra generated by the open sets in $\cX$.

\subsection{Proof of Theorem \ref{thm:entrywise_estimator}}\label{sec:proof_main}

\subsubsection{Useful Lemmas}
\begin{theorem} [Hoeffding's inequality for bounded random variables {\cite[Theorem 2.2.6]{vershyninHDP}}] \label{thm:hoeffdingbounded}
Let $X_1, \ldots, X_n$ be independent random variables. Assume that $X_i \in [m_i, \hspace{1pt} M_i]$ for every $i \in [n].$ Then, for any $t >0,$ we have
\[
P \left ( \sum_{i = 1}^n \left ( X_i - \E X_i \right ) \geq t \right ) \leq \exp \left ( -\dfrac{2t^2}{\sum_{i = 1}^n (M_i - m_i)^2} \right ).
\]
\end{theorem}

\begin{lemma} \label{lem:ratio_bound}
Let $a, b, c, d \in \RR$ such that $b \neq 0$, $b + d \neq 0,$ and $\left \vert \dfrac{a}{b} \right \vert  \leq 1.$ Then
\[
\left \vert \dfrac{a + c}{b + d} - \dfrac{a}{b} \right \vert \leq \dfrac{\left \vert c \right \vert + \left \vert d \right \vert}{\left \vert b + d \right \vert }.
\]
\end{lemma}
\begin{proof}
\begin{align*}
\left \vert \dfrac{a + c}{b + d} - \dfrac{a}{b} \right \vert & = \left \vert \dfrac{b c - a d}{b ( b + d ) } \right \vert \\
& \leq \left \vert \dfrac{c}{b + d} \right \vert + \left \vert \dfrac{a}{b} \right \vert \left \vert \dfrac{d}{b + d} \right \vert \\
& \leq  \dfrac{\left \vert c \right \vert + \left \vert d \right \vert}{\left \vert b + d \right \vert }.
\end{align*}
\end{proof}

\subsubsection{Proof of Main Result}

\begin{proof}
We seek to apply Lemma \ref{lem:ratio_bound} with $a = w^2 \theta_i \theta_j,$ $b = w,$ $c = (A \bOne \bOne^T A)_{ij}  - w^2 \theta_i \theta_j,$ and $d =  \bOne^T A \bOne - w$.
We begin by establishing fundamental bounds on the absolute ``errors'' $\left \vert (A \bOne \bOne^T A)_{ij}  - w^2 \theta_i \theta_j \right \vert $ and $\left \vert \bOne^T A \bOne - w \right \vert$. Then, we show that these error bounds are ``small'' via a concentration inequality for sums of independent random variables. Finally, we apply Lemma \ref{lem:ratio_bound} to yield the result.

Fix $i,j \in [n]$ and write $E_{ij} \equiv \left ( \dfrac{A \bOne \bOne^T A}{ \bOne^T A \bOne} \right )_{ij} - w \theta_i \theta_j.$ To this end, we begin with fundamental bounds on the ``error'' terms $\bOne^T A \bOne - w$ and $(A \bOne \bOne^T A)_{ij}  - w^2 \theta_i \theta_j$. We have
\[
(A \bOne \bOne^T A)_{ij} = \sum_{l = 1}^n a_{il} \sum_{m = 1}^n a_{mj},
\]
and
\begin{align*}
 \bOne^T A \bOne & = \sum_{p = 1}^n \sum_{q = 1}^n a_{pq} \\
& = 2 \sum_{p = 1}^n \sum_{q = p+1}^n a_{pq} + \sum_{r = 1}^n a_{rr}.
\end{align*}
Then,
\begin{align} 
    & \left \vert (A \bOne \bOne^T A)_{ij}  - w^2 \theta_i \theta_j \right \vert  = \left \vert \sum_{l = 1}^n a_{il} \sum_{m = 1}^n a_{mj} - w^2 \theta_i \theta_j \right \vert \nonumber \\
    & = \left \vert  \left ( \sum_{l = 1}^n a_{il}  - w \theta_i \right )  \left (\sum_{m = 1}^n a_{mj} - w \theta_j \right ) \right. \nonumber  \\
    & \left. - 2 w^2 \theta_i \theta_j + w \theta_j \sum_{l = 1}^n a_{il} + w \theta_i \sum_{m = 1}^n a_{mj} \right \vert \nonumber \\
    & \leq \left \vert \sum_{l = 1}^n a_{il}  - w \theta_i \right \vert  \left \vert \sum_{m = 1}^n a_{mj} - w \theta_j \right \vert \nonumber \\
    & +  w \theta_j \left \vert \sum_{l = 1}^n a_{il} - w \theta_i \right \vert + w \theta_i \left \vert \sum_{m = 1}^n a_{mj} - w \theta_j \right \vert, \label{eq:tri1}
\end{align}
and
\begin{align}\label{eq:tri2}
& \left \vert  \bOne^T A \bOne - \sum_{p = 1}^n \sum_{q = 1}^n w \theta_p \theta_q \right \vert \leq  \left \vert \sum_{r = 1}^n a_{rr} - \sum_{r = 1}^n \theta^2_r w \right \vert & \\
& + 2 \left \vert  \sum_{p = 1}^n \sum_{q = p+1}^n a_{pq} - \sum_{p = 1}^n \sum_{q = p+1}^n \theta_p \theta_q w \right \vert.
\end{align}
By Theorem \ref{thm:hoeffdingbounded} (Hoeffding's inequality), for any $t \geq 0$ we have
\begin{enumerate}[label=(\roman*)]
\item $P \left (  \left \vert \displaystyle \sum_{l = 1}^n a_{il} -  \theta_i  w \right \vert \geq t \right ) \leq \exp \left ( -\dfrac{2 t^2}{n} \right )$;
\item $P \left ( \left \vert \displaystyle \sum_{m = 1}^n a_{mj} -  \theta_j w \right \vert  \geq t \right ) \leq \exp \left ( -\dfrac{2 t^2}{n} \right )$;
\item \begin{align*} & P \left (  \left \vert \displaystyle \sum_{p = 1}^n \sum_{q = p+1}^n a_{pq} - \displaystyle \sum_{p = 1}^n \sum_{q = p+1}^n \theta_p \theta_q w \right \vert  \geq t \right ) \\
& \leq \exp \left (- \dfrac{2 t^2}{n^2 - n} \right ) \leq \exp \left (- \dfrac{2 t^2}{n^2} \right ); \end{align*}
\item $P \left (  \left \vert \displaystyle \sum_{r = 1}^n a_{rr} - \displaystyle  \sum_{r = 1}^n \theta^2_r w \right \vert  \geq t \right ) \leq \exp \left ( -\dfrac{2 t^2}{n} \right ),$
\end{enumerate}  
where in (i) and (ii) we used the normalization condition $\sum_{l = 1}^n \theta_l = 1$ (cf. Definition \ref{def:DCSBM}).
Combining \eqref{eq:tri1}, (i), and (ii), we find from union bound and DeMorgan's law that 
\begin{align} 
& P \left ( \left \vert (A \bOne \bOne^T A)_{ij}  - w^2 \theta_i \theta_j \right \vert \leq t^2 + w ( \theta_i + \theta_j ) t \right ) \nonumber \\
&\geq 1 - 2 \exp \left ( -\dfrac{2 t^2}{n} \right ), \qquad \forall t \geq 0. \label{eq:prob_bound_num}
\end{align}
Similarly, from (iii), (iv), and \eqref{eq:tri2}, it follows that
\begin{align} \label{eq:prob_bound_denom}
P \left ( \left \vert \bOne^T A \bOne -  w \right \vert \leq 3 t \right ) & \geq 1 - \left ( \exp \left (- \dfrac{2 t^2}{n^2} \right ) + \exp \left ( -\dfrac{2 t^2}{n} \right ) \right ) \nonumber \\
& \geq 1 -  2 \exp \left (- \dfrac{2 t^2}{n^2} \right ), \qquad \forall t \geq 0.
\end{align}
Note that $w = \sum_{p = 1}^n \sum_{q = 1}^n w \theta_p \theta_q \geq n^2 \epsilon,$ thus, \eqref{eq:prob_bound_denom} implies
\begin{align}
P \left ( \left \vert \bOne^T A \bOne -  w \right \vert \leq \dfrac{w}{2} \right ) & \geq 1 - 2 \exp \left (- \dfrac{w^2}{18n^2} \right ) \nonumber \\
& \geq 1 - 2 \exp \left (- \dfrac{n^2 \epsilon^2}{18} \right ). \label{eq:coarse_bound_w}
\end{align}
Assuming $\left \vert w - \bOne^T A \bOne \right \vert \leq \frac{w}{2}$ and, thus, $\bOne^T A \bOne > 0$, Lemma \ref{lem:ratio_bound}\footnote{with $a = w^2 \theta_i \theta_j,$ $b = w,$ $c = (A \bOne \bOne^T A)_{ij}  - w^2 \theta_i \theta_j,$ and $d =  \bOne^T A \bOne - w$} yields
\begin{align}
\left \vert  E_{ij} \right \vert  & \leq \dfrac{\left \vert (A \bOne \bOne^T A)_{ij} - w^2 \theta_i \theta_j \right \vert + \left \vert \bOne^T A \bOne - w \right \vert}{\bOne^T A \bOne} \nonumber \\
& = \dfrac{\left \vert (A \bOne \bOne^T A)_{ij} - w^2 \theta_i \theta_j \right \vert + \left \vert \bOne^T A \bOne - w \right \vert}{\left \vert w - ( w - \bOne^T A \bOne) \right \vert } \nonumber \\
& \leq \dfrac{\left \vert (A \bOne \bOne^T A)_{ij} - w^2 \theta_i \theta_j \right \vert + \left \vert \bOne^T A \bOne - w \right \vert}{\left \vert w - \left \vert w - \bOne^T A \bOne \right \vert  \right \vert } \nonumber \\
& \leq  2 \dfrac{\left \vert (A \bOne \bOne^T A)_{ij} - w^2 \theta_i \theta_j \right \vert + \left \vert \bOne^T A \bOne - w \right \vert}{\left \vert w  \right \vert } \nonumber\\
& \leq 2 \dfrac{\left \vert (A \bOne \bOne^T A)_{ij} - w^2 \theta_i \theta_j \right \vert + \left \vert \bOne^T A \bOne - w \right \vert}{n^2 \epsilon}. \label{eq:ratio_bound}
\end{align}
Noting that $w \theta_p = w \theta_p \sum_{q = 1}^n \theta_q \leq n (1 - \epsilon) $ for all $p \in [n]$ and letting $t>0$ be arbitrary, it follows that
\begin{align*}
& P \left ( \left \vert  E_{ij} \right \vert \leq 2 \dfrac{\sqrt{t} (2 n^{3/2} (1 - \epsilon)) +  4tn}{n^2 \epsilon} \right ) \\
& \geq P \left ( \left \vert E_{ij} \right \vert \leq 2 \dfrac{\left ( \sqrt{tn} \right )^2 + w (\theta_i + \theta_j) \sqrt{tn}  + 3tn}{n^2 \epsilon} \right ) \\
& \geq 1 - 2 \left ( 2 e^{-2 t^2} + e^{-\frac{n^2 \epsilon^2}{18}} \right ),
\end{align*}
where in the last step we combined \eqref{eq:prob_bound_num}, \eqref{eq:prob_bound_denom}, \eqref{eq:coarse_bound_w}, and \eqref{eq:ratio_bound}.
Finally, for any $t > 0$,
\begin{align*}
2 \dfrac{\sqrt{t} (2 n^{3/2} (1 - \epsilon)) +  4tn}{n^2 \epsilon} & = \dfrac{4 (1 - \epsilon) \sqrt{t}  n^{-1/2} + 8 t n^{-1}}{\epsilon} \\
& \leq \dfrac{8}{\epsilon} \left ( \sqrt{t} n^{-1/2} + t n^{-1} \right ).
\end{align*}

\end{proof}

\bibliographystyle{plain}
\bibliography{asilomar}	

\end{document}